%% file: paper.tex
\documentclass[12pt]{article} \usepackage{fullpage}
\usepackage{amssymb}
\usepackage{amsmath}
\usepackage{amsthm}
\newcommand{\nR}{\mathbb{R}}

\usepackage{algorithm}
\usepackage{algpseudocode}
\usepackage{nicefrac}
\usepackage{url}
\usepackage{color}
\usepackage[american]{babel}
\usepackage{graphicx}
\usepackage{subfigure}

\newtheorem{theorem}{Theorem}
\newtheorem{corollary}{Corollary}
\newtheorem{definition}{Definition}

\newtheorem{observation}{Observation}

\newtheorem{lemma}{Lemma}

\newcommand{\leaveout}[1]{}

\renewcommand{\medskip}{\smallskip}
\renewcommand{\int}{{\ensuremath{\rm int\,}}}

\usepackage{lineno}

\date{}
\title{On triangulating $k$-outerplanar graphs}
\author{Therese Biedl\thanks{David R.~Cheriton School of Computer
Science, University of Waterloo, Waterloo, Ontario N2L 1A2, Canada.
Supported by NSERC and the Ross and Muriel Cheriton Fellowship.}
}

\begin{document}

\maketitle
\begin{abstract}
A $k$-outerplanar graph is a graph that can be drawn in the plane
without crossing such that after $k$-fold removal of the vertices on the 
outer-face there are no vertices left.  In this paper, we study how
to triangulate a $k$-outerplanar graph while keeping its outerplanarity
small.  Specifically, we show that not all $k$-outerplanar graphs can
be triangulated so that the result is $k$-outerplanar, but they can
be triangulated so that the result is $(k+1)$-outerplanar.
\end{abstract}

\section{Introduction}

A {\em planar graph} is a graph $G=(V,E)$ that can be drawn in
the plane without crossing.  Given such a drawing $\Gamma$, the
{\em faces} are the connected pieces of $\nR^2-\Gamma$; the
unbounded piece is called the {\em outer-face}.  A planar
drawing can be described by
giving for each vertex the clockwise order of edges at it, and
by saying which edges are incident to the outer-face; we call this
a {\em combinatorial embedding}.

Assume that a planar drawing $\Gamma$ has been fixed.
Define $L_1$ to be the vertices incident to the outer-face, and
define $L_i$ for $i>1$ recursively to be the
vertices on the outer-face of the planar drawing obtained when 
removing the vertices in $L_1,\dots,L_{i-1}$.  We call  $L_i$ (for $i\geq 1$)
the {\em $i$th onion peel}  of drawing $\Gamma$.
A graph is called {\em $k$-outerplanar} if it has a planar drawing that
has most $k$ onion peels.  The {\em outer-planarity} of a planar graph $G$ is
the smallest $k$ such that $G$ is $k$-outerplanar.

A {\em triangulated graph} is a planar graph for which all faces (including
the outer-face) are triangles.  A {\em triangulated disk} is a planar
graph for which the outer-face is a simple cycle and all inner faces
(i.e., faces that are not the outer-face) are triangles.
It is well-known that any planar graph can be {\em triangulated},
i.e., we can add edges to it without destroying planarity so that
it becomes triangulated.

Sometimes it is of interest to triangulate a planar graph while
maintaining other properties.  For example, any planar graph without
separating triangles can be triangulated without creating separating
triangles \cite{BKK97}, with the exception of graphs with a universal
vertex.
Any planar graph can be triangulated
so that the maximum degree increases by at most a constant
\cite{KB97}.  Any planar graph $G$ can be triangulated such
that the result has treewidth at most $\max\{3,tw(G)\}$ \cite{BR13}.
Also, following the proof of Heawood's 3-color
theorem \cite{Hea1898}, one can easily show that any 3-colorable
planar graph can be made triangulated by adding edges and vertices
such that the result is 3-colorable. 

In this paper, we investigate whether a planar graph
can be triangulated without changing its outer-planarity.  We
show first that this is not true.  For example, a 4-cycle has
outer-planarity 1, but the only way to triangulate it is to create
$K_4$, which has outer-planarity 2.  (We give more complicated
examples for higher outer-planarity in Section~\ref{se:not_possible}.)
However, if we are content with ``only'' converting the graph to a
triangulated disk, then it is
always possible to do so without increasing the outer-planarity
(see Section~\ref{se:triangulatedDisk}).
In consequence, any $k$-outerplanar graph can be triangulated
so that its outer-planarity is at most $k+1$.  
In Section~\ref{se:appl} we use our triangulations to give a
different proof of the well-known result \cite{Bod98}
that $k$-outerplanar graphs have treewidth at most $3k-1$. 

\section{Triangulating $k$-outerplanar graphs}
\label{se:not_possible}

In this section, we show that not all planar graphs can be
triangulated while maintaining the outer-planarity.

\begin{theorem}
For any $k\geq 1$, there exists a triangulated disk $G$ with $O(k)$ vertices
that is $k$-outerplanar, but any triangulation
of $G$ has outer-planarity at least $k+1$.
\end{theorem}
\begin{proof}
For $k=1$, the graph $K_4$ with one edge deleted is a suitable example.
For $k>1$, we first define an auxiliary graph $T_i$ as follows.  $T_1$ consists
of a single triangle $t_1$.  
$T_i$, for $i>1$, is obtained by taking a triangle $t_i$
and inserting a copy of $T_{i-1}$ inside it; then add a 6-cycle between
triangles $t_i$ and $t_{i-1}$.  In other
words, $T_i$ consists of $i$ nested triangles. 
Clearly graph $T_i$ is 3-connected and 
has $i$ onion peels if $t_i$ is the outer-face.
See Figure~\ref{fig:outpl_increases} (left).

We now define graph $G$ to consist of four copies of $T_k$, in the embedding
with $t_k$ on the outer-face, and connect them so that the outer-face contains
two vertices of each copy of $t_k$.  The inner faces ``between'' the four
copies of $T_k$ are triangulated arbitrarily.  
See Figure~\ref{fig:outpl_increases} (right).
Notice that the first and second onion
peel will contain (in each copy of $T_k$) all vertices of $t_k$ and $t_{k-1}$.
Therefore the $i$th onion
peel (for $2\leq i\leq k$) contains $t_{k-i}$ 
and hence $G$ is $k$-outerplanar.  It is also a triangulated disk
and has $12k$ vertices.

\begin{figure}[ht]
\hspace*{\fill}
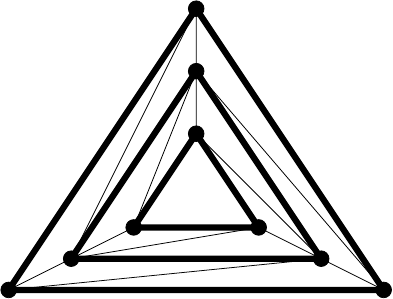
\hspace*{\fill}
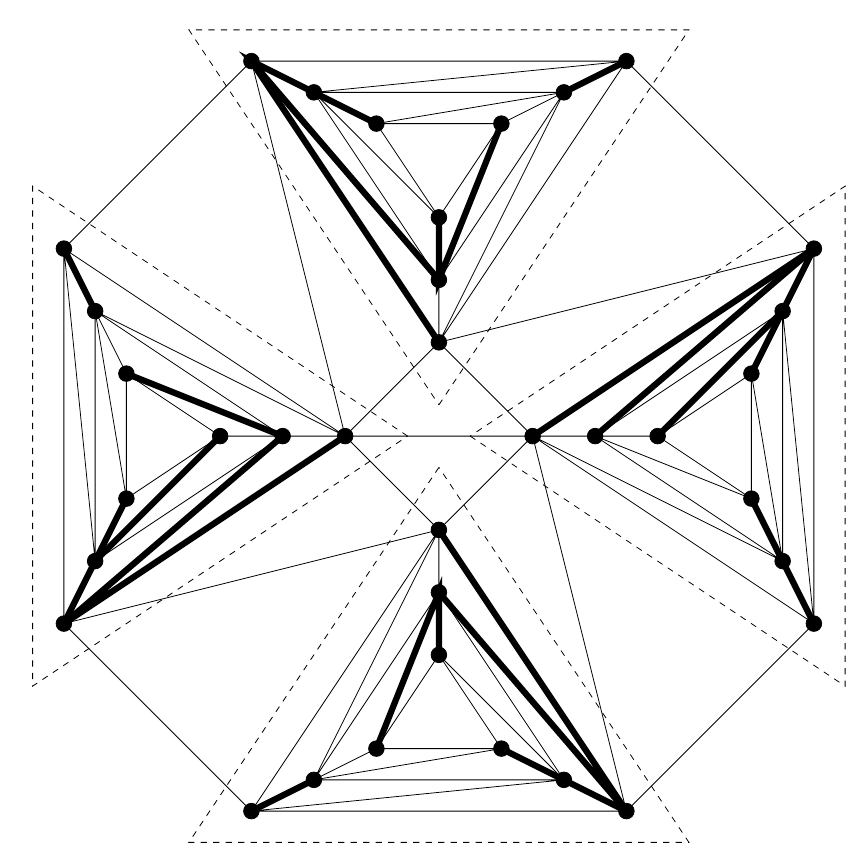
\hspace*{\fill}
\caption{(Left) Graph $T_3$.  (Right) 
A 3-outer planar graph which cannot be triangulated and stay 3-outerplanar.
Thick edges indicate an outer-face-rooted spanning forest of height 2 (defined
formally in Section~\ref{se:triangulated_disk}.}
\label{fig:outpl_increases}
\end{figure}

Now let $G'$ be any triangulation of $G$.  Since there are three vertices
on the outer-face of $G'$, there exists one copy $C$ of $T_k$ that does not
have any vertex on the outer-face.  In consequence (since $T_k$ is
3-connected), the embedding of $C$ induced by $G'$ must have $t_k$ as its
outer-face.  The first onion peel of $G'$ contains no vertex of $C$.
In consequence, at least $k+1$ onion peels are required before all vertices
of $C$ are removed, and the outer-planarity of $G'$ is at least $k+1$.
\end{proof}

\section{Converting to triangulated disks}
\label{se:triangulated_disk}
\label{se:triangulatedDisk}

In this section, we aim to show that we can triangulate inner faces
without increasing the outer-planarity.  To our knowledge, this result
was not formally described in the literature before (though Lemma 3.11.1
in \cite{Bod86} has many of the crucial steps for it.)
From now on, let $G$ be a $k$-outerplanar graph with the planar
embedding and outer-face fixed such that it has onion peels
$L_1,L_2,\dots,L_k$.  We first compute a special spanning forest of $G$
(after adding some edges).  We need some preliminary results

\begin{observation}
\label{obs:neighbourOnPreviousOnionPeel}
If $v\in L_i$ (for some $i>1$),
then some incident face of $v$ contains vertices in $L_{i-1}$.
\end{observation}
\begin{proof}
Since $v$ is in $L_i$ and not in $L_{i-1}$, it is not on the outer-face
of the graph $H_{i-1}$ induced by $L_{i-1}\cup L_i \cup L_{i+1}\cup \dots$.
Therefore all incident faces of $v$ (in $H_{i-1}$)
are inner faces.  But since $v$ {\em is} on the outer-face after
deleting $L_{i-1}$, at least one of its incident faces merges with
the outer-face when removing $L_{i-1}$.  Therefore at least one
incident face of $v$ contains a vertex from $L_{i-1}$.
\end{proof}

\begin{observation}
\label{obs:makeNeighbourToPreviousOnionPeel}
We can add edges 
(while maintaining planarity) such that every vertex in $L_i$, $i>1$
has a neighbor in $L_{i-1}$.
\end{observation}
\begin{proof}
Add edges in any inner face $f$ as follows:  Let $w$ be the vertex of $f$
contained in the onion peel with smallest index among all vertices of $f$
(breaking ties arbitrarily.)  For any vertex $v\neq w$ of $f$, add
an edge $(v,w)$ if it did not exist already.  Clearly this maintains
planarity since all new edges can be drawn inside face $f$.

By Observation~\ref{obs:neighbourOnPreviousOnionPeel}, every vertex $v\in L_i$
(for $i>1$) had an incident face $f_v$ that contained a vertex in
$L_{i-1}$.  When applying the above procedure to face $f_v$ some vertex
$w$ in $L_{i-1}$ is made adjacent to $v$, unless $(w,v)$ already was an edge.
Either way, afterwards $v$ has the neighbor $w\in L_{i-1}$.
\end{proof}

A {\em spanning forest} of $G$ is a subgraph that contains all vertices
of $G$ and has no cycles.  We say that a spanning forest is 
{\em outer-face-rooted}
if every every connected component of it contains
exactly one vertex on the outer-face.    We say that an outer-face-rooted 
spanning forest $F$
has {\em height h} if every vertex $v$ has distance (in $F$) at most $h$
to an outer-face vertex.
See also Figure~\ref{fig:outpl_increases}.

\begin{lemma}
\label{le:koutplAddEdgesSpanningTree}
Let $G$ be a $k$-outerplanar graph.  The we can add edges to
$G$ (while maintaining planarity) such that $G$ has an outer-face-rooted 
spanning forest of height at most $k-1$.
\end{lemma}
\begin{proof}
First add edges as in Observation~\ref{obs:makeNeighbourToPreviousOnionPeel}.
Now any vertex $v$ in $L_i$, $i\geq 1$ has distance at most $i-1$
from some vertex in $L_1$:  This holds by definition for $i=1$, and holds
by induction for $i>1$, since 
vertex $v$ has a neighbor $w$ in $L_{i-1}$ and
$w$ has distance at most $i-2$ to some vertex in $L_1$.  

Now perform a breadth-first search, starting at all the vertices 
on the outer-face $L_1$.  The resulting breadth-first
search tree $F$ (which is a forest, since we start with multiple vertices)
has one component for each outer-face vertex.  Since 
breadth-first search computes distances from its start-vertices, 
each vertex has distance at most $k-1$ from a root of $F$ and so
$F$ has height at most $k-1$.
\end{proof}

\begin{lemma}
\label{lem:spanningTreeSmallHeightkoutpl}
Let $G$ be a planar graph that (for some fixed planar embedding
and outer-face) has an outer-face-rooted spanning
forest $F$ of height $k-1$.
Then $G$ is $k$-outerplanar.
\end{lemma}
\begin{proof}
Root each connected component $T$ of $F$ at the vertex on the outer-face.
Removing the outer-face $L_1$ then removes the root of each tree $T$.  
After the roots have been removed,
all their children appear on the outer-face of what remains. 
So all children of the roots
are in $L_2$.  (There may be other vertices in $L_2$ as well.)
Continuing the argument
shows that the vertices at distance $i$ from the roots
are in onion peel $L_{i+1}$ or in one of
earlier onion peels $L_1,\dots,L_i$.  
Therefore $G$ has at most $k$ non-empty onion peels and it is $k$-outerplanar.
\end{proof}

\begin{theorem}
\label{th:koutplMakeInnerTriangulated}
Any $k$-outerplanar graph $G$ can be converted into a $k$-outerplanar
triangulated disk by adding edges. 
\end{theorem}
\begin{proof}
Add edges to $G$ (while maintaining
planarity) until it has an outer-face-rooted spanning forest $F$ of
height $k-1$ (Lemma~\ref{le:koutplAddEdgesSpanningTree}).
While the outer-face is disconnected,
add an edge between two vertices on the outer-face of different
connected components.
While the outer-face has a vertex $v$ that appears on it multiple times, add
an edge between two neighbors of $v$ on the outer-face.
 Finally, add more edges to $G$ (with the standard
techniques for triangulating) until all interior faces are triangles.
Note that none of these edges additions
removes any vertex from the outer-face.  So we end with a triangulated
disk $D$ whose outer-face vertices are the same as the ones on $G$.
In particular, $F$ is an outer-face-rooted spanning forest of $D$ as
well, and it still has height $k-1$.
By Lemma~\ref{lem:spanningTreeSmallHeightkoutpl} $D$ is
$k$-outerplanar as desired.
\end{proof}

\begin{corollary}
\label{th:koutplMakeTriangulated}
Any $k$-outerplanar graph $G$ can be triangulated 
such that the result has outer-planarity at most $k+1$.
\end{corollary}
\begin{proof}
First convert $G$ into a triangulated disk $D$ that is
$k$-outerplanar.  Now pick one vertex $r$ on the outer-face
of $D$ that has only two neighbors on the outer-face on $r$.
This exists because the outer-face induces a 2-connected 
outer-planar graph; such graphs have a degree-2 vertex. 
Make $r$ adjacent to all other vertices on the outer-face.
Clearly the result $G'$ is a triangulated graph.  
Also, if
$L_0',L_1',\dots$  are the onion peels of $G'$, then
$r\in L_0'$, any neighbors of $r$ (and in particular therefore
all of $L_1$) is in $L_0'\cup L_1'$, and by induction any
vertex in $L_i$ is in $L_0'\cup \dots \cup L_i'$.  Therefore $G'$
has at most $k+1$ onion peels as desired.
\end{proof}

\section{Treewidth of $k$-outerplanar graphs}
\label{se:app}
\label{se:appl}

It is well-known that any $k$-outerplanar graph
has treewidth at most $3k-1$ \cite{Bod88,Bod98} and
this bound is tight \cite{KT09}.  
(We will not review the definition
of treewidth here, since we will only use the closely related
concept of branchwidth.)  This has important algorithmic
consequences: many (normally NP-hard) problems
can be solved in polynomial time on $k$-outerplanar graphs, 
which allows for a PTAS for many problems in planar graphs 
(see Baker \cite{Baker94}), or for
solving graph isomorphism and related problems
efficiently in planar graphs (see Eppstein \cite{Eppstein99}.)

The proof in \cite{Bod98} is non-trivial and in particular
requires first converting the $k$-outerplanar graph $G$ into
a $k$-outerplanar graph $H$ with maximum degree 3 such that
$G$ is a minor of $H$.  A detailed discussion
(and analysis of the linear-time complexity to find the tree
decomposition) is given in \cite{Kat13}.
A second, different, proof can be derived from Tamaki's theorem
\cite{Tam03} that shows that the branchwidth of a graph is
bounded by the radius of the face-vertex-incidence graph.  But
this proof is not straightforward either, as it requires detours
into the medial graph  and the carving width.

Our result on triangulating $k$-outerplanar graphs, in conjunction
with some results of Eppstein concerning tree decompositions of
graphs with small diameter \cite{Eppstein99}, allows for a different
(and in our opinion simpler) proof that every $k$-outerplanar
graph has treewidth at most $3k-1$.  We explain this in the
following.

We first need to define a closely related concept, the {\em branchwidth}.

\begin{definition}
A {\em branch decomposition} of a graph $G$ is a tree $T$ that has 
maximum degree 3, together with an injective
assignment of the edges of $G$ to the leaves of $T$.
In such a branch decomposition,
a vertex $v$ of $G$ is said to {\em cross} an arc $a$ of $T$ if
two incident edges of $v$ are assigned to leaves in two different
components of $T-a$.  The branch decomposition is said to have
{\em width} $w$ if any arc $a$ of $T$ is crossed by at most $w$ vertices.
The {\em branchwidth} of a graph $G$ is the minimum width of a branch
decomposition of $G$.
\end{definition}

The following lemma relates the branchwidth of a planar graph $G$ to the height
of an outer-planar-rooted spanning forest $F$ of $G$.  It is strongly inspired
by Lemma 4 of \cite{Eppstein99} (which in turn was inspired by \cite{Baker94}):

\begin{lemma}
Let $G$ be a triangulated disk with an outer-face-rooted spanning forest $F$
of height $h-1$.  Then $G$ has branchwidth at most $2h$.
\end{lemma}
\begin{proof}
Let $G^*$ be
the dual graph of $G$.  Let $T^*$ be a subgraph of $G^*$ defined
as follows:  $T^*$ contains all vertices of $G^*$ (= faces of $G$),
except for the outer-face of $G$.  It also contains the duals of all
edges of $E$ that are not in $F$ and not on the outer-face of $G$.
See also Figure~\ref{fi:RadiusBranchwidth}(left).

We claim that $T^*$ is a tree.  This can be seen as follows.  Define
$F^+$ to be the subgraph of $G$ formed by the edges of $F$, as well as all
but one edge on the outer-face.  Since $F$ is an outerface-rooted forest, 
$F^+$ is a spanning tree of $G$.  By the well-known tree-co-tree result
(\cite{Tutte84}, p.289) therefore the duals of the edges not in $F^+$ form
a spanning tree $T^+$ of the dual graph.  The outer-face-vertex is a leaf 
in $T^+$ by definition of $F^+$.  Deleting this leaf from $T^+$
yields exactly $T^*$, which therefore is a tree.

%

We will use $T^*$ (with some additions) as the tree for the branch
decomposition.  See also Figure~\ref{fi:RadiusBranchwidth}.
A node of $T^*$ will be called {\em face-node}
and denoted $n(f)$ if it corresponds to the inner face $f$ of $G$.
Let $T_1$ be the tree obtained from $T^*$
by subdividing each arc $a$ of $T^*$ with an {\em arc-node} $n(a)$.
Let $T_2$ be the tree obtained from $T_1$
by adding an {\em edge-node} $n(e)$ for every edge $e$ of $G$.  
If the dual edge $e^*$ of $e$ is an arc of $T^*$, then
make $n(e)$ adjacent to the arc-node $n(e^*)$; note
that $n(e^*)$ had degree 2 before and is used for exactly one $n(e)$,
so it has degree 3 now.  If the dual edge of $e$ is not in $T^*$,
then either $e$ is on the outer-face or $e$ belongs to $F$.
In both cases, pick an inner face $f$ incident to $e$ and make
$n(e)$ adjacent to $n(f)$.  Notice that in $T_2$ node $n(f)$
has at most one incident arc for each edge of $f$,
therefore $n(f)$ has degree at most 3.  

\begin{figure}[ht]
\hspace*{\fill}
\includegraphics[width=70mm,page=3]{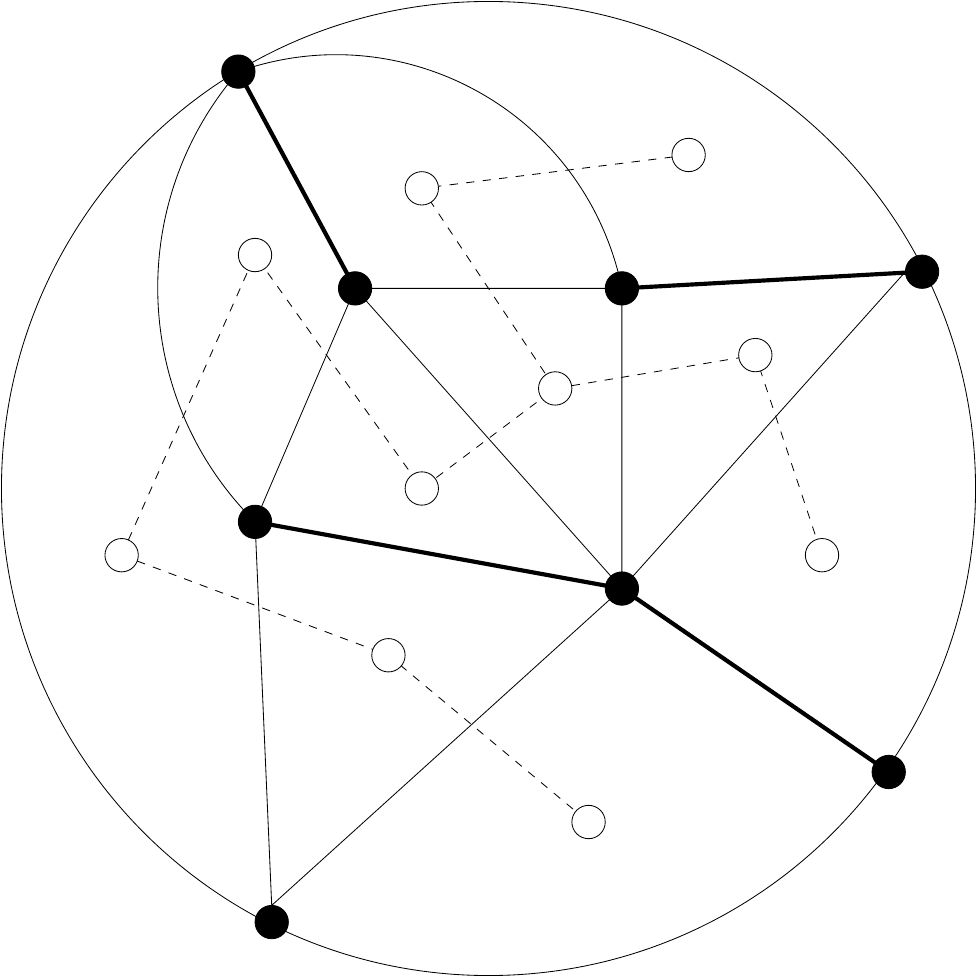}
\hspace*{\fill}
\includegraphics[width=70mm,page=6]{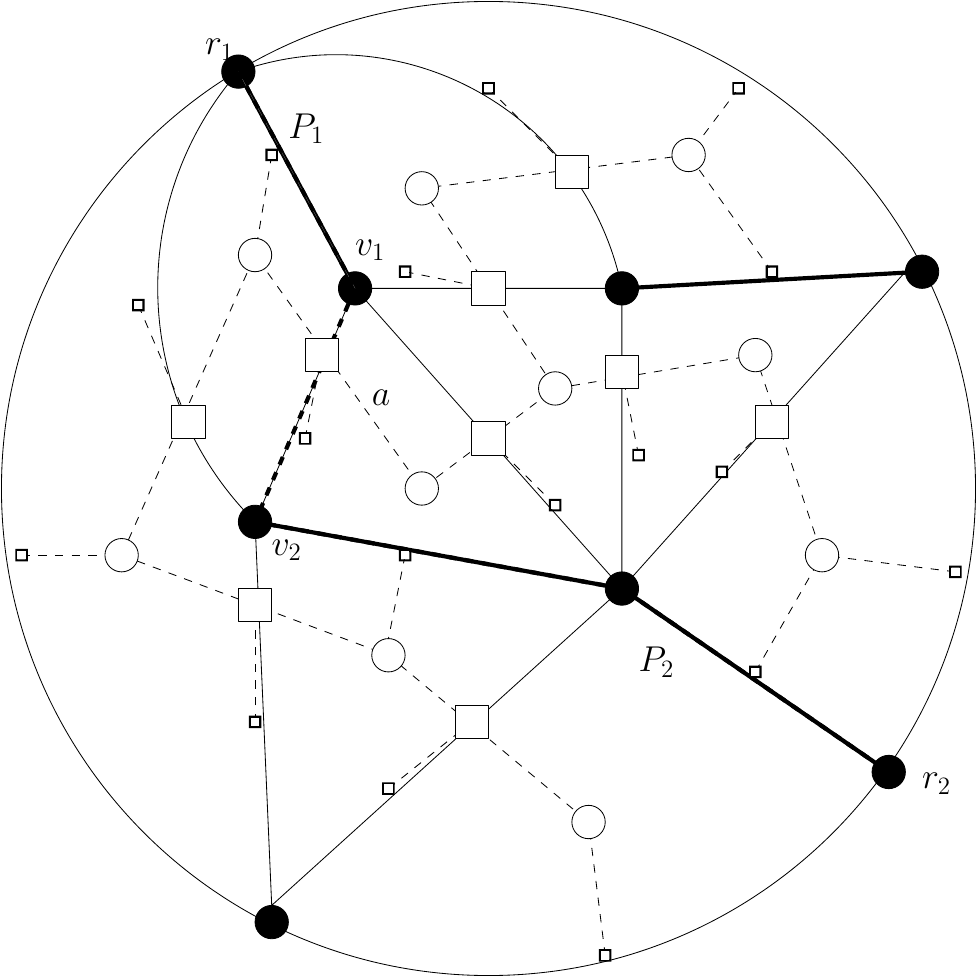}
\hspace*{\fill}
\caption{A branch decomposition obtained from an outer-face-rooted
spanning forest.  (Left) Tree $T^*$.  $G$ is solid; forest-edges are thick.
Face-nodes are white circles and arcs of $T^*$ are dashed.
(Right) Tree $T_2$. Arc-nodes are white squares,
edge-node $n(e)$ is a small square drawn near edge $e$.  The labels
indicate the path $r_1-v_1-v_2-r_2$ from outer-face to outer-face
defined by an arc $a$ of $T_2$. }
\label{fi:RadiusBranchwidth}
\end{figure}

We use tree $T_2$ for the branch decomposition and assign
edge $e$ of $G$ to node $n(e)$.  We have already
argued that $T_2$ has maximum degree 3, so it is a branch decomposition,
and it only remains to analyze its width. 
Let $a$ be an arc of $T_2$.  If $a$ is incident to a node $n(e)$
of $T_2$, then only the vertices of $e$ can cross $a$, so at most
$2 \leq 2h$ vertices cross $a$.  If $a$ is not incident to a node
$n(e)$, then it has the form $(n(f),n(e^*))$ for some inner face
$f$ of $G$ and some edge $e=(v_1,v_2)$ that is incident to $f$
and does not belong to $F$.

If $v_1$ and $v_2$ are in different connected components of $F$, then
for $j=1,2$, let $P_j$ be the path from $v_j$ to the 
outer-face vertex $r_j$ in $v_j$'s component of $F$.
Observe that $P_1$ and $P_2$ are disjoint, and therefore 
$P_1\cup \{e\} \cup P_2$ is a path from outer-face
to outer-face that splits the inner faces of $G$ into two parts, namely,
the two parts corresponding to the two connected components of $T_2-a$.
Any vertex that has incident edges in both those connected components
hence must be on $P_1\cup \{e\} \cup P_2$.  But $P_1$ and $P_2$ contain
at most $h-1$ edges each, so there are at most $2h$ vertices that cross $a$.
Similarly, if $v_1$ and $v_2$ are in the same connected component of $F$,
then let $P$ the path from $v_1$ to $v_2$ in $F$, and observe that $P\cup \{e\}$
forms a cycle that separates the two components of $T_2-a$.  Since
$P$ contains at most $2h-2$ edges, in this case at most $2h-1$ vertices
cross $a$.  

So this branch decomposition has width at most $2h$ as desired.
\end{proof}

Since $tw(G)\leq \max\{1,\lfloor\frac{3}{2}bw(G)\rfloor-1\}$ for the treewidth $tw(G)$
and branchwidth $bw(G)$ of a graph \cite{RS-GMX}, we therefore have:

\begin{corollary}
\label{co:spanningTreeTreewidth}
Let $G$ be a triangulated disk with a outer-face-rooted spanning forest $F$
of height $h-1$.  Then $G$ has treewidth at most $3h-1$.
\end{corollary}

Since adding edges does not decrease the treewidth, therefore by Lemma
\ref{le:koutplAddEdgesSpanningTree} we have:

\begin{corollary}
\label{co:koutplTreewidth}
Any $k$-outerplanar graph has treewidth at most $3k-1$.
\end{corollary}

Following the steps of our proof, it is easy to see that the
branch decomposition of width $2k$ can be found in linear time,
and from it, a tree decomposition of width $3k-1$ is easily obtained
by following the proof in \cite{RS-GMX}.

\bibliographystyle{plain}
\bibliography{../bib/journal,../bib/full,../bib/gd,../bib/papers}


\end{document}

%% file: T3.pdf_t
\begin{picture}(0,0)%
\includegraphics{T3.pdf}%
\end{picture}%
\setlength{\unitlength}{1973sp}%
\begingroup\makeatletter\ifx\SetFigFont\undefined%
\gdef\SetFigFont#1#2#3#4#5{%
  \reset@font\fontsize{#1}{#2pt}%
  \fontfamily{#3}\fontseries{#4}\fontshape{#5}%
  \selectfont}%
\fi\endgroup%
\begin{picture}(3766,2866)(6218,-6744)
\put(7951,-5911){\makebox(0,0)[lb]{\smash{{\SetFigFont{10}{12.0}{\rmdefault}{\mddefault}{\updefault}{\color[rgb]{0,0,0}$t_1$}%
}}}}
\put(8926,-5086){\makebox(0,0)[lb]{\smash{{\SetFigFont{10}{12.0}{\rmdefault}{\mddefault}{\updefault}{\color[rgb]{0,0,0}$t_3$}%
}}}}
\end{picture}%

%% file: outpl_increases.pdf_t
\begin{picture}(0,0)%
\includegraphics{outpl_increases.pdf}%
\end{picture}%
\setlength{\unitlength}{1973sp}%
\begingroup\makeatletter\ifx\SetFigFont\undefined%
\gdef\SetFigFont#1#2#3#4#5{%
  \reset@font\fontsize{#1}{#2pt}%
  \fontfamily{#3}\fontseries{#4}\fontshape{#5}%
  \selectfont}%
\fi\endgroup%
\begin{picture}(8130,8193)(3886,-7075)
\put(10501,839){\makebox(0,0)[lb]{\smash{{\SetFigFont{10}{12.0}{\rmdefault}{\mddefault}{\updefault}{\color[rgb]{0,0,0}$T_k$}%
}}}}
\put(10501,-6961){\makebox(0,0)[lb]{\smash{{\SetFigFont{10}{12.0}{\rmdefault}{\mddefault}{\updefault}{\color[rgb]{0,0,0}$T_k$}%
}}}}
\put(12001,-661){\makebox(0,0)[lb]{\smash{{\SetFigFont{10}{12.0}{\rmdefault}{\mddefault}{\updefault}{\color[rgb]{0,0,0}$T_k$}%
}}}}
\put(3901,-661){\makebox(0,0)[lb]{\smash{{\SetFigFont{10}{12.0}{\rmdefault}{\mddefault}{\updefault}{\color[rgb]{0,0,0}$T_k$}%
}}}}
\end{picture}%

%% file: paper.bbl
\begin{thebibliography}{10}

\bibitem{Baker94}
B.~Baker.
\newblock Approximation algorithms for {NP}-complete problems on planar graphs.
\newblock {\em J. ACM}, 41(1):153--180, 1994.

\bibitem{BKK97}
T.~Biedl, G.~Kant, and M.~Kaufmann.
\newblock On triangulating planar graphs under the four-connectivity
  constraint.
\newblock {\em Algorithmica}, 19(4):427--446, 1997.

\bibitem{BR13}
T.~Biedl and L.E.~Ruiz Velazquez.
\newblock Drawing planar 3-trees with given face areas.
\newblock {\em Computational Geometry: Theory and Applications},
  46(3):276--285, 2013.

\bibitem{Bod86}
H.~Bodlaender.
\newblock Classes of graphs with bounded tree-width.
\newblock Technical Report RUU-CS-86-22, Rijskuniversiteit Utrecht, 1986.

\bibitem{Bod88}
H.~Bodlaender.
\newblock Planar graphs with bounded treewidth.
\newblock Technical Report RUU-CS-88-14, Rijksuniversiteit Utrecht, 1988.

\bibitem{Bod98}
Hans~L. Bodlaender.
\newblock A partial {\it k}-arboretum of graphs with bounded treewidth.
\newblock {\em Theor. Comput. Sci.}, 209(1-2):1--45, 1998.

\bibitem{Eppstein99}
David Eppstein.
\newblock Subgraph isomorphism in planar graphs and related problems.
\newblock {\em J. Graph Algorithms Appl.}, 3(3), 1999.

\bibitem{Hea1898}
P.J. Heawood.
\newblock On the four-color map theorem.
\newblock {\em Quart. J. Pure Math.}, 29, 1898.

\bibitem{KT09}
Frank Kammer and Torsten Tholey.
\newblock A lower bound for the treewidth of $k$-outerplanar graphs.
\newblock Technical Report 2009-07, Universit{\"a}t Augsburg, 2009.

\bibitem{KB97}
Goos Kant and Hans~L. Bodlaender.
\newblock Triangulating planar graphs while minimizing the maximum degree.
\newblock {\em Inf. Comput.}, 135(1):1--14, 1997.

\bibitem{Kat13}
Ioannis Katsikarelis.
\newblock Computing bounded-width tree and branch decompositions of
  $k$-outerplanar graphs.
\newblock {\em CoRR}, abs/1301.5896, 2013.

\bibitem{RS-GMX}
Neil Robertson and P.~D. Seymour.
\newblock Graph minors. {X}. {O}bstructions in tree-decompositions.
\newblock {\em J. Combin. Theory Ser. B}, 52:153--190, 1991.

\bibitem{Tam03}
Hisao Tamaki.
\newblock A linear time heuristic for the branch-decomposition of planar
  graphs.
\newblock In {\em European Symposium on Algorithms (ESA'03)}, volume 2832 of
  {\em Lecture Notes in Computer Science}, pages 765--775. Springer, 2003.

\bibitem{Tutte84}
W.T. Tutte.
\newblock {\em Graph Theory}.
\newblock Addison-Wesley, 1984.

\end{thebibliography}
